%% file: paper.tex
\pdfoutput=1

\documentclass[preliminary,copyright,creativecommons,english]{eptcs}
\providecommand{\event}{QPL 2018}
\usepackage{underscore}           

\usepackage[T1]{fontenc}
\usepackage[latin9, utf8]{inputenc}
\usepackage{babel}
\usepackage{mathtools}
\usepackage{amsmath}
\usepackage{amsfonts}
\usepackage{amsthm}
\usepackage{amssymb}
\usepackage{xspace}
\usepackage{booktabs}
\usepackage{stmaryrd}
\usepackage{centernot}
\usepackage{graphics}
\usepackage{epstopdf}
\usepackage{cleveref}
\usepackage{algorithm}
\usepackage{algpseudocode}
\usepackage{colortbl}
\usepackage{qcircuit}
\usepackage{semantic}
\usepackage{underscore}
\usepackage{subcaption}
\usepackage[sort,nocompress]{cite}
\usepackage{tikz}
\usetikzlibrary{matrix,patterns,hobby,arrows,decorations.markings}


\input{macros.tex}

\makeatletter
\newcommand{\pushright}[1]{\ifmeasuring@#1\else\omit\hfill$\displaystyle#1$\fi\ignorespaces}
\newcommand{\pushleft}[1]{\ifmeasuring@#1\else\omit$\displaystyle#1$\hfill\fi\ignorespaces}
\makeatother

\tikzset{
    node style bb/.style={rectangle,fill=black!20},
    node style ba/.style={red,rectangle,fill=black!20},
    node style aa/.style={red}
}

\title{Towards Large-scale Functional Verification of  \\ Universal Quantum Circuits}

\author{Matthew Amy
	\institute{Institute for Quantum Computing and David R. Cheriton School of Computer Science \\ University of Waterloo, Canada}
	\email{meamy@uwaterloo.ca}
}

\def\titlerunning{Towards Large-scale Functional Verification of Universal Quantum Circuits}
\def\authorrunning{M. Amy}

\begin{document}

\maketitle

\begin{abstract}
We introduce a framework for the formal specification and verification of quantum circuits based on the Feynman path integral. Our formalism, built around exponential sums of polynomial functions, provides a structured and natural way of specifying quantum operations, particularly for quantum implementations of classical functions. Verification of circuits over all levels of the Clifford hierarchy with respect to either a specification or reference circuit is enabled by a novel rewrite system for exponential sums with free variables. Our algorithm is further shown to give a polynomial-time decision procedure for checking the equivalence of Clifford group circuits. We evaluate our methods by performing automated verification of optimized Clifford+$T$ circuits with up to 100 qubits and thousands of $T$ gates, as well as the functional verification of quantum algorithms using hundreds of qubits. Our experiments culminate in the automated verification of the Hidden Shift algorithm for a class of Boolean functions in a fraction of the time it has taken recent algorithms to simulate.
\end{abstract}

\section{Introduction}

Verification is a fundamental aspect of modern electronic design. Without a high level of assurance that a circuit design conforms to a particular specification, chip makers stand to lose hundreds of millions of dollars when their product is inevitably recalled. The consequences in the quantum computing realm aren't quite as clear, as the largely software-like nature of quantum circuits alleviates much of the risk associated with design flaws. On the other hand, quantum resource analyses, which typically vary wildly between compilers \cite{IARPAQCS}, are currently being used to assess and guide real security policies \cite{glrs16, adgmps16}, so it is highly desirable to attain some degree of assurance that these resource analyses are indeed correct.

Due to the absence of large, universal quantum computers and the inherent difficulty of simulating quantum circuits, testing is generally not a viable option for verification. By contrast, various methods of formal verification have been developed for quantum circuits and programs, including equivalence checking \cite{w09, y10, agn14}, diagrammatic methods \cite{dl13, gd17}, model checkers \cite{gnp08, z16}, program logics \cite{y12} and formal proof \cite{rpz18}. However, two questions remain: how can the intended effect of a quantum program be specified in a clear, human readable and verifiable way, and how can we scale automated verification to large circuits?

Typical \emph{functional} verification methods -- verification of the precise input-output relation -- either verify equivalence against a simpler circuit or diagrammatic implementation (e.g., \cite{w09, y10, gd17}), or a matrix representation such as a unitary or superoperator (e.g., \cite{rpz18}). With either approach, errors can creep in \emph{on the specification side}, as both circuit and matrix presentations can be difficult for humans to write and understand. Moreover, in the former case it is assumed that a certified implementation exists in the first place, and in the latter case the matrix either requires exponential space to write and store, or is left abstract \cite{rpz18}, relying on structural proofs which are generally not suitable for verifying heavily optimized circuits.

In this work we propose a novel framework for the formal specification and functional verification of unitary (i.e., measurement-free) quantum circuits over a universal gate set -- specifically, the Clifford group extended with $Z$-axis rotations taken from the Clifford hierarchy \cite{gc99}. Our framework is built around Richard Feynman's \emph{path integral} technique, which has been used recently to prove results in complexity theory \cite{dhmhno05, m17}, and to perform circuit simulation \cite{bg16, kps17} and optimization \cite{amm14, am16, aam17}. Specifically, we develop a concrete representation of quantum operators as \emph{path-sums} -- exponential sums of basis states over a finite set of Boolean \emph{path variables}. Our path-sums directly coincide with the standard mathematical presentation of common quantum circuits and algorithms (e.g., \cite{nc00}), and further allow the direct use of classical functions, which can themselves be tested or otherwise verified, to formally specify quantum operations.

To verify quantum circuits, we give a computable, compositional semantics of quantum circuits as path-sums. We show that over Clifford+$R_k$ circuits for any fixed $k$, this interpretation is efficiently computable and compact. We then present a reduction system for path-sums which iteratively reduces the number of path variables until a (non-unique) normal form is reached. Our reduction system together with an efficient initial transformation is complete for Clifford group circuits, giving a polynomial-time equivalence checking algorithm. Experimentally, we use our reduction system to perform the automated verification of optimized Clifford+$T$ circuits, as well as Clifford+$R_k$ implementations of various quantum algorithms against formal specifications as path-sums for up to 200 qubits.


\paragraph{Preliminaries}

We work in the strictly unitary picture of quantum computing \cite{nc00} -- that is, quantum computations are modelled by unitary operators on a complex vector space of dimension $2^n$. While we do not consider measurements, we allow qubit initialization, corresponding to partial isometries on a complex vector space. We denote the computational basis vectors as $\ket{\x}$ for binary strings $\x=x_1x_2\dots x_n\in\F^n$.

A circuit is defined as a sequence of quantum gates applied to individual qubits. We primarily consider three quantum gates:
\[
	H = \frac{1}{\sqrt{2}}\begin{pmatrix} 1 & 1 \\ 1 & -1 \end{pmatrix}, \quad
	R_k = \begin{pmatrix} 1 & 0 \\ 0 & e^{\frac{2\pi i}{2^k}}  \end{pmatrix}, \quad \text{and} \quad
	\cnot = \begin{pmatrix} 1 & 0 & 0 & 0 \\ 0 & 1 & 0 & 0 \\ 0 & 0 & 0 & 1 \\ 0 & 0 & 1 & 0 \end{pmatrix}.
\]
For $k\geq 1$, all three gates lie in the $k^{\text{th}}$ level of the \emph{Clifford hierarchy}, denoted $\mathcal{C}_k$, where $\mathcal{C}_1$ is the Pauli group and $\mathcal{C}_k=\{U|U\mathcal{C}_1 U^\dagger \subseteq \mathcal{C}_{k-1}\}$. Two important cases are the \emph{Clifford group} ($\mathcal{C}_2$) and \emph{Clifford+$T$} ($\mathcal{C}_3$). While for $k\leq 3$ the above gates suffice to generate $\mathcal{C}_k$, it is not generally known whether $\mathcal{C}_k=\langle H, R_k, \cnot\rangle$.

Much of our formalism involves polynomial representations of pseudo-Boolean functions -- functions from $\F^n$ into some set $S$. In particular, we are interested in pseudo-Boolean functions into the ring of \emph{dyadic fractions} $\D=\{\frac{a}{2^b} | a, b \in \Z\}$, which correspond uniquely to multilinear polynomials in $\D_M[\x]=\D[\x]/\langle x_i^2 - x_i\rangle$. In our context the ring of dyadic fractions arises from the phase factors of $R_k$ gates, and are needed to precisely represent the quantum Fourier transform.

\section{The path-sum framework}\label{sec:sop}

\begin{figure}
\centerline{\scalebox{1}{
    \begin{tikzpicture}
    \node at (7.4,3) {$B$};
    \node at (0,0)    {$A$};
    \draw [thick, decoration={markings, mark=at position 0.7 with {\arrow{triangle 60}}}, postaction=decorate] (0.2 ,0.2) to 
    	[curve through ={([out angle=45, in angle=200]3.4,2)  . . (5.8,0.7) .. (6.5,2.3)  }] (7.2,2.8);
    \draw [thick, decoration={markings, mark=at position 0.4 with {\arrow{triangle 60}}}, postaction=decorate] (0.2 ,0.2) to 
    	[curve through ={([out angle=45, in angle=180]3.4,1.4)  . . (4.2,1.2)  }] (7.2,2.8);
    \draw [thick, decoration={markings, mark=at position 0.4 with {\arrow{triangle 60}}}, postaction=decorate] (0.2 ,0.2) to 
    	[curve through ={(1.5,3)  . . (4.3,2.2) . . (6.5,3)  }] (7.2,2.8);
    \end{tikzpicture}
}}
\caption{The paths of a particle from point $A$ to $B$.}
\label{fig:pathintegral}
\end{figure}

The path-sum dates back to Feynman and the path integral formulation of quantum mechanics \cite{fh65}. In a general sense, the idea is to describe the amplitude of a particular state (say, of a particle) by an integral over all possible paths leading to that state. \Cref{fig:pathintegral} shows the trajectories of a particle moving from states $A$ to $B$ -- in the path integral formulation, the final amplitude is described as the sum of the amplitudes of each path. The output amplitudes of a quantum circuit, as a quantum mechanical system, can likewise be described as the sum over all trajectories of the system. However, as quantum gates are typically modelled as operators on a finite dimensional Hilbert space, a discrete sum rather than integral is typically used \cite{dhmhno05, bvr08, m17, kps17}.

We can describe a path-sum abstractly as a discrete set of \emph{paths} $S\subseteq \F^m$, together with an amplitude function $\phi$ and state transformation $f$ representing the operator
\[
	U : \ket{\x} \mapsto \sum_{\y\in S}\phi(\x, \y)\ket{f(\x, \y)}.
\]
In this form, the path-sum is not particularly useful as a computational representation, as the representations of $\phi$ and $f$ are not fixed -- indeed $\phi$ itself may be a unitary matrix with $\phi(\x, \y)$ indexing a particular entry. Instead, we fix a concrete representation based on multivariate polynomials which suffices to exactly represent most interesting quantum operations.

\begin{definition}[path-sum]\label{def:sop}
	An $n$-qubit \emph{path-sum} $\xi$ consists of
	\begin{itemize}
		\item an \emph{input signature} $\ket{\x=x_1x_2\cdots x_n}$  where each $x_i$ is a (distinct) variable or Boolean constant,
		\item a \emph{phase polynomial} $P\in\D_M[\x,\y]$ over input variables $\x$ and \emph{path variables} $\y=y_1y_2\dots y_m$, and
		\item an \emph{output signature} $\ket{f(\x,\y)=f_1(\x,\y)\cdots f_n(\x,\y)}$ where each $f_i\in\F[\x,\y]$ is a Boolean polynomial.
	\end{itemize}
	The \emph{associated operator} of a path-sum is the partial linear map $U_{\xi}$ where
	\[
		U_{\xi} : \ket{\x} \mapsto \frac{1}{\sqrt{2^m}}\sum_{\y\in\F^m} e^{2\pi i P(\x,\y)}\ket{f(\x,\y)}.
	\]
\end{definition}

We say a path variable is \emph{internal} if it does not appear in the output signature. Our presentation is inspired by descriptions of quantum operators in mathematical texts \cite{nc00, klm07}, and as such we write a path-sum informally by the action of its associated operator.
By an abuse of notation, we use $\ket{\x}$ to refer to either an input signature or an arbitrary Boolean vector corresponding to an input signature. 

\begin{example}
Path-sum representations of common quantum gates and circuits are listed below:
\begin{align*}
	T: &\ket{x} \mapsto e^{2\pi i\frac{x}{8}}\ket{x} \\
	H:&\ket{x}\mapsto\frac{1}{\sqrt{2}}\sum_{y\in\F}e^{2\pi i\frac{xy}{2}}\ket{y} \\
	\textsf{Toffoli}_n: &\ket{x_1x_2\cdots x_n} \mapsto \ket{x_1x_2\cdots (x_n\oplus \textstyle \prod_{i=1}^{n-1}x_i)} \\
	\textsf{Adder}_n: &\ket{\x}\ket{\y}\ket{\0}\mapsto \ket{\x}\ket{\y}\ket{\x+ \y} \\
	\textsf{QFT}_n: &\ket{\x}\mapsto \frac{1}{\sqrt{2^n}}\sum_{\y\in\F^n}e^{2\pi i \frac{\int{\x\cdot \y}}{2^n}}\ket{\y}
\end{align*}
Addition and multiplication of Boolean vectors are interpreted as integer operations at the bit level. In the \textsf{QFT} above, $\int{\x \cdot \y}$ denotes the integer value of $\x \cdot \y$. For any classical function $f$, we can lift the polynomial representation of $f$ to a quantum operator via the path-sum $\ket{\x}\ket{\0}\mapsto \ket{\x}\ket{f(\x)}$. Note that the polynomial representation of a classical function may grow exponentially large, as in the case of addition. A practical implementation of path-sums as a specification language would include a classical sub-language, along with a verified translation from such programs into Boolean polynomials.
\end{example}

As a unitary or partial isometry may admit many distinct path-sum representations, we define an equivalence between path-sums with the same associated operator.

\begin{definition}[equivalence] 
 Two path-sums $\xi_1, \xi_2$ are \emph{equivalent}, denoted $\xi_1 \equiv \xi_2$, if and only if their associated operators are equal -- that is, $U_{\xi_1} = U_{\xi_2}$.
\end{definition}

An additional point to note is that non-isometric path-sums are possible in our framework, as for instance $\ket{x}\mapsto\ket{0}$ is a valid path-sum. In this work we are concerned only with the unitary circuit model and by extension isometric path-sums, hence we define a notion of well-formedness for path-sums.

\begin{definition}[well-formed]
	A path-sum is well-formed if its associated operator is a (partial) isometry.
\end{definition}

In practice, well-formedness is only an issue when writing path-sums directly as specifications, and our verification methods work even when a path-sum is not guaranteed to be well-formed. We leave it as a question for future research to determine methods for checking well-formedness of path-sums.

\subsection{Compositions of path-sums}

As with quantum circuits, path-sums may be composed both \emph{vertically} and \emph{horizontally} -- that is, composed in parallel with another path-sum on a distinct subsystem or in sequence on the same subsystem, respectively. Vertical composition is defined in the obvious way -- concatenating the inputs and outputs then adding the phase polynomials with appropriate renaming -- but horizontal composition requires more care.

Intuitively, as path-sums symbolically describe mappings between linear combinations of basis vectors, we can compose the output $\ket{f(\x, \y)}$ of one path-sum with the input $\ket{\x'}$ of another by substituting each input value $x_i'$ with the corresponding output $f_i(\x,\y)$. For instance, we can compute the composition of $\ket{x_1x_2x_3}\mapsto\ket{x_1(x_1\oplus x_2)x_3}$ followed by $\ket{x_1'x_2'x_3'}\mapsto\ket{x_1'x_2'(x_2'\oplus x_3')}$ by substituting $x_2'$ with $x_1\oplus x_2$:
\[
	\ket{x_1x_2x_3}\mapsto\ket{x_1(x_1\oplus x_2)(x_1\oplus x_2\oplus x_3)}.
\] However, this presents a problem when the path-sum on the left (i.e. right-to-left composition) is a partial isometry, as we may end up composing a variable $f_i(\x, \y)=x_j$ with a constant state $x_i'=b$ for some $b\in\F$, effectively post-selecting on $x_j=b$. For this reason we require that only \emph{compatible}\footnote{Determining compatibility is at least as hard as detecting whether an ancilla is \emph{clean} and is hence non-trivial in general. For the verification tasks we consider this is not an issue, as in practice we only compose path-sums with unitary operators.} signatures are composed; in particular, an output $\ket{f(\x,\y)}$ is compatible with an input $\ket{\x'}$ if and only if for every $i$, either $x_i'$ is a variable or $x_i'=b=f_i(\x, \y)$.

When the left-most path-sum has a non-zero phase polynomial, substitutions may extend to the phase. As the phase and output polynomials are defined over different rings ($\D$ and $\F$, respectively), when substituting a variable with a Boolean polynomial in the phase we first need to lift it into a \emph{functionally equivalent} polynomial over $\D$. For instance, for all $x,y\in\F$, $\frac{1}{4}\left( x \oplus y\right) = \frac{1}{4}x + \frac{1}{4}y - \frac{1}{2}xy$. We define the \emph{lifting} of a Boolean polynomial $P$ to a polynomial $\lift{P}\in\D_M[\x]$ recursively by
\begin{align*}
	\lift{\x^\alpha}&=\x^\alpha, \\  \lift{P + Q} &= \lift{P} + \lift{Q} - 2\lift{PQ},
\end{align*}
where $\x^\alpha=x_1^{\alpha_1}x_2^{\alpha_2}\cdots x_n^{\alpha_n}$ for $\alpha\in\F^n$ is a multi-index, and the first equation uses the natural inclusion of $\F$ in $\D$. It can be easily verified that the lifting of a Boolean polynomial preserves its action on elements of $\F$.

\begin{lemma}\label{lem:poly}
For any Boolean-valued polynomial $P$ and all $\x\in\F^n$, $\lift{P}(\x) = P(\x) \mod 2$.
\end{lemma}

We can now formally define the functional composition of path-sums.

\begin{definition}{(sequential composition)} \\
The \emph{sequential composition} of two compatible path-sums 
\[
	U_{\xi} : \ket{\x} \mapsto \frac{1}{\sqrt{2^{m}}}\sum_{\y\in\F^{m}} e^{2\pi i P(\x,\y)}\ket{f(\x,\y)}, \qquad 
	U_{\xi'} : \ket{\x'} \mapsto \frac{1}{\sqrt{2^{m'}}}\sum_{\y'\in\F^{m'}} e^{2\pi i P'(\x,\y)}\ket{f'(\x',\y')},
\]
denoted $\xi'\circ \xi$, is given by
\[
	U_{\xi'\circ \xi} : \ket{\x} \mapsto \frac{1}{\sqrt{2^{m+m'}}}\sum_{\y\in\F^{m+m'}}
		e^{2\pi i \left(P + P'[y_i \gets y_{i+m}][x'_i \gets \lift{f_i}]\right)(\x,\y)}\ket{\left(f'[x'_i\gets f_i]\right)(\x,\y)},
\]
where $P[x \gets Q]$ for polynomials $P,Q$ over some ring $R$ denotes the substitution of $x$ with $Q$ in $P$.
\end{definition}

\begin{proposition}
For any well-formed, compatible path-sums $\xi, \xi'$, $\xi'\circ\xi$ is also well formed. Moreover, 
\[U_{\xi'\circ \xi} = U_{\xi'}U_{\xi}.\]
\end{proposition}

\begin{remark}
A useful property of path-sums is that they unify structurally equivalent circuits without resorting to string diagrams, which can be difficult to reason about in automated ways \cite{bgksz16}. By this we mean that circuits which are equivalent up to symmetric monoidal laws are \emph{strictly equal} in the path-sum picture. For instance, the bifunctoriality law and the naturality of SWAP, stated respectively as the equivalences
\[
	\Qcircuit @C=1em @R=.4em {
 		& \gate{f} & \qw  & \qw \\
 		& \qw & \gate{g} & \qw
	} 
	\raisebox{-0.9em}{$\equiv$}
	\Qcircuit @C=1em @R=.4em {
 		& \qw & \gate{f}  & \qw \\
 		& \gate{g} & \qw & \qw
	}
	\qquad \qquad
	\Qcircuit @C=1em @R=.9em {
 		& \gate{f} & \qswap  & \qw \\
 		& \qw & \qswap \qwx & \qw
	} 
	\raisebox{-0.9em}{$\equiv$}
	\Qcircuit @C=1em @R=.9em {
 		& \qswap & \qw  & \qw \\
 		& \qswap \qwx & \gate{f} & \qw
	}
\]
are both equality in the path-sum framework. While much progress has been made towards computational methods for diagrammatic reasoning \cite{dl13, bgksz16, cdkw16, gd17}, our framework allows us to use standard algebraic tools (e.g., rewriting) without explicitly managing structural laws.

Along with unifying the representation of \emph{structurally} equivalent circuits, path-sums further unify many \emph{semantic} equivalences of quantum circuits -- particularly allowing the long-distance cancellation of phase gates applied to the same logical state \cite{amm14}. In contrast, matrix representations unify \emph{all} equivalences between unitaries, at the expense of exponential space. Path-sums hence provide an intermediary model, where many equivalences are ``modded out'' yet still generally remain efficiently representable as we show next.
\end{remark}

\subsection{Path-sums as a circuit semantics}

As path-sums admit both a symmetric tensor product and functional composition, we can give a simple compositional path-sum semantics of measurement-free quantum circuits. Given a path-sum representation of each gate in a basis $\mathcal{B}$ and their inverses, we can define the path-sum interpretation of a circuit over $\mathcal{B}$ as the composition of each gate. In particular, we give a path-sum interpretation to the Clifford+$R_k$ basis $\{H, \cnot, R_k\}$ for $k>0$.

\begin{definition}{(Clifford+$R_k$ path-sum)} \\
The path-sum interpretation of an $n$-qubit circuit $C$ over $\{H, \cnot, R_k\}$, denoted $\sop{C}$, is defined as follows:
\begin{alignat*}{2}
	\sop{H} &= \ket{x}\mapsto\frac{1}{\sqrt{2}}\sum_{y\in\{0,1\}}e^{2\pi i\frac{xy}{2}}\ket{y} \\
	\sop{R_k} &= \ket{x}\mapsto e^{2\pi i\frac{x}{2^k}}\ket{x} \\
	\sop{R_k^\dagger} &= \ket{x}\mapsto e^{2\pi i\frac{-x}{2^k}}\ket{x} \\
	\sop{\cnot} &= \ket{x_1x_2}\mapsto \ket{x_1(x_1\oplus x_2)}\\
	\sop{C_1; C_2} &= \sop{C_2} \circ \sop{C_1}.
\end{alignat*}
We leave the appropriate vertical compositions implicit.
\end{definition}

\begin{proposition}
For any circuit $C$ over $\{H, \cnot, R_k\}$ with unitary matrix $U_C$, we have $U_{\sop{C}}=U_C$.
\end{proposition}

As a composition of linear Boolean functions, it can trivially be observed that each of the outputs of a canonical path-sum is linear. Moreover, its phase polynomial has degree at most $k$. To show this, we first introduce the notion of the \emph{order} of a polynomial in $\D_M$ which gives a more precise characterization of the phase polynomials over a fixed level of the Clifford hierarchy. Note that without loss of generality we can restrict our attention to phase polynomials with coefficients in $\D/\Z$ since $e^{2\pi i}=1$.

\begin{definition}
The \emph{order} of a term $\frac{a}{2^b} \x^{\alpha}$ where $a$ is co-prime to $2$ and $\alpha\in\F^n$ is $b + |\alpha| - 1$. The order of a polynomial $P\in\D_M[\x]$, denoted $\ord{P}$, is the maximum order of all terms in $P$.
\end{definition}

\begin{example}
	\[
		\ord{\frac{1}{2}} = 0, \qquad
		\ord{\frac{1}{2}x_1 + \frac{1}{2}x_2} = 1, \qquad
		\ord{\frac{1}{2^3}x_2 + \frac{1}{2}x_1x_2x_3} = 3
	\]
\end{example}

An important fact, shown below, is that order is non-increasing with respect to substitution of linear Boolean polynomials.

\begin{lemma}\label{lem:order}
	Let $P\in\D_M[\x]$, and let $Q\in\F[\x]$ be a linear polynomial. Then for any $x_i$,
	\[
		\ord{P[x_i\gets \lift{Q}]} \leq \ord{P}
	\]
\end{lemma}
\begin{proof}
Suppose $Q=\sum_{j\in S} x_j$ for some set $S$. It is easy to verify that
\[
	\lift{\sum_{j\in S} x_j} = \sum_{S'\subseteq S}(-2)^{|S'| - 1}\prod_{j\in S'}x_j.
\]
Substituting $\lift{Q}$ in for $x_i$ we see that for any term $\frac{a}{2^b} \x^{\alpha}$ in $P$ such that $\alpha_i=1$,
\begin{align*}
	\ord{\frac{a}{2^b} \x^{\alpha}[x_i\gets \lift{Q}]} 
		&=\max_{S'\subseteq S}\ord{\frac{a2^{|S'| - 1}}{2^b} \x^{\alpha}[x_i\gets \prod_{j\in S'}x_j]} \\
		&\leq\max_{S'\subseteq S} b-|S'|+|\alpha| + |S'| - 1 \\
		&= \ord{\frac{a}{2^b} \x^{\alpha}}.
\end{align*}
\end{proof}

Intuitively, as the output function of a Clifford+$R_k$ path-sum is strictly linear, composing Clifford+$R_k$ path-sums does not increase the order of the phase polynomial. Moreover, the path-sum interpretation of each gate over $\{H, \cnot, R_k\}$ has a phase polynomial of order at most $k$ and maximum denominator $2^k$, hence we obtain the following result.

\begin{proposition}
The phase polynomial of a (canonical) Clifford+$R_k$ path-sum has degree at most $k$.
\end{proposition}

\begin{corollary}\label{cor:size}
The path-sum interpretation of an $n$-qubit Clifford+$R_k$ circuit $C$ has size polynomial in the volume of $C$ ($n\cdot |C|$) and can be computed in polynomial time.
\end{corollary}

\paragraph{On representations of the phase polynomial}
While the representation of the phase as a multilinear polynomial is indeed polynomial in the size of the circuit, at higher levels of the Clifford hierarchy (i.e. large $k$) the degree of the polynomial can become prohibitively large. Even for the standard Clifford+$T$ gate set, the path-sum of a circuit requires space cubic in the volume of the circuit \cite{am16}. In practice this makes verification of some larger circuits difficult.

The phase polynomial could instead be represented in \emph{linear space for any $k$} by its \emph{Fourier expansion} \cite{od14, aam17}. This however complicates the process of verification as the Fourier expansion is not necessarily unique modulo integer multiples \cite{aam17}. A possible compromise would be to store the Fourier expansion normally, and generate the multilinear form for small subsets on demand. 

\section{A calculus for path-sums}\label{sec:rewrite}
The verification question we're generally concerned with is \emph{given a circuit $C$ and path-sum $\xi$, is $\sop{C}\equiv \xi$?}. From an automated perspective it is simpler to instead check that the path-sum \emph{miter} \cite{y10} $\sop{C^\dagger}\circ \xi$ is the identity transformation. In either case, we need a method of efficiently establishing equivalence. To that end, in this section we present a system of reduction rules for path-sums. A key feature of our calculus is that the reduction rules strictly decrease the number of path variables, producing a (not necessarily unique) normal form in polynomial time.

\subsection{Overview}

Our calculus operates by reducing the number of paths when sets of paths interfere in recognizable ways which we call \emph{interference patterns}. As an illustration, consider the identity circuit $HH$. Computing its canonical path-sum we get
\[
	HH : \ket{x} \mapsto \frac{1}{\sqrt{2^2}}\sum_{y_1,y_2\in\F}e^{2\pi i\frac{xy_1+y_1y_2}{2}}\ket{y_2}.
\]
To see that the above path-sum is equal to the identity, we can first expand the exponential sum on the right by the values of the internal path variable $y_1$:

\begin{align*}
	\frac{1}{\sqrt{2^2}}\sum_{y_1,y_2\in\F}e^{2\pi i\frac{xy_1+y_1y_2}{2}}\ket{y_2}
		&= \frac{1}{\sqrt{2^2}}\sum_{y_2\in\F}(1 + e^{2\pi i\frac{x+y_2}{2}})\ket{y_2}
\end{align*}
Since $e^\pi i = -1$, it can be observed that if $x+y_2 = 0 \mod 2$, the two paths corresponding to $y_1=0$ and $y_2=1$ \emph{constructively} interfere, whereas if $x+y_2=1\mod 2$ they \emph{destructively} interfere. As $\F= x\oplus \F = \{x, 1\oplus x\}$ for any $x\in\Z$, we can rewrite the sum over $x\oplus \F$ and explicitly calculate the interference on either path:
\begin{align*}
	\frac{1}{\sqrt{2^2}}\sum_{y_2\in x\oplus \F}(1 + e^{2\pi i\frac{x+y_2}{2}})\ket{y_2} 
		&=\frac{1}{2}(1 + e^{2\pi i\frac{x+x}{2}})\ket{x} + 
			\frac{1}{2}(1 + e^{2\pi i\frac{x+1+x}{2}})\ket{1 \oplus x} \\
		&= \frac{2}{2}\ket{x} + \frac{0}{2}\ket{1\oplus x} \\
		&= \ket{x}
\end{align*}

The reasoning above applies to any situation where an internal path variable $y_i$ only appears with coefficients taken from the Boolean subgroup $\{0,\frac{1}{2}\}$ of $\D/\Z$, as the two branches of $y_i$ are identical, except that $y_i=1$ path picks up a multiplicative factor of $-1$ whenever the quotient of $P/y_i$ is odd. Specifically, it can be shown that
\begin{align*}
	\frac{1}{\sqrt{2^{m+1}}}\sum_{y_0\in\F}\sum_{\y\in\F^m}e^{2\pi i\left (\frac{1}{2} y_0Q(\x, \y) + R(\x, \y)\right)}\ket{f(\x, \y)}
		= \frac{1}{\sqrt{2^{m-1}}}\sum_{\y\in\F^m, Q(\x, \y) = 0\mod 2}e^{2\pi iR(\x, \y)}\ket{f(\x, \y)}
\end{align*}
Note that the polynomial $Q(\x, \y)$ is Boolean-valued, as otherwise the $y_0=1$ path can pick up values not in $\{1, -1\}$. In practice, we only perform such reductions when the restricted sum can be reified by solving $Q(\x, \y) = 0\mod 2$ for some path variable, as we did above with $y_2=x$.

\subsection{Reduction rules}

\Cref{fig:rewrite} gives the rules of our calculus, presented as algebraic rewrite rules on exponential sums for convenience and applied to path-sums in the obvious way. We write $\xi\reduces \xi'$ to denote that $\xi$ reduces to $\xi'$, and denote by $\reducestrans$ the transitive closure of $\reduces$. For all rules, $y_0$ is an internal path variable, quotients $Q$ are Boolean-valued and whenever $y_i\gets Q$, $y_i$ does not appear in $Q$. For the \textsf{[Case]} rule, both $y_i$ and $y_j$ are internal.

The rules were developed by translating known circuit identities into path-sums, then minimizing the identities to obtain simple interference patterns which 1) strictly reduce the number of path variables, and 2) can be efficiently matched. What we found was that most common Clifford+$T$ equalities reduce to a small set of rules -- in particular, the \textsf{[HH]} rule derived from the equality $HH=I$ as described above is sufficient for the vast majority of path-sum reductions. The \textsf{[$\omega$]} rule arises from the identity $(SH)^3 = e^{\frac{2\pi i}{8}}I$, and the final rule \textsf{[Case]} is a specific case distinction needed to prove the 2-qubit Clifford+$T$ identity $\left(\cnot(X\otimes T) \textsf{controlled-}H(X\otimes T^\dagger)\right)^2$ \cite{bs16}. The \textsf{[Elim]} rule only appears to simplify the presentation of \textsf{[HH]} as well as in some contexts specific to verification which we describe later.

\begin{figure}
\resizebox{\textwidth}{!}{\begin{minipage}{\linewidth}
\begin{alignat*}{3}
	&\frac{1}{\sqrt{2^{m+2}}}\sum_{y_0\in\F}\sum_{\y\in\F^m}e^{2\pi iP(\x, \y)}\ket{f(\x,\y)}
		&&\reduces\frac{1}{\sqrt{2^{m}}}\sum_{\y\in\F^m}e^{2\pi iP(\x, \y)}\ket{f(\x,\y)} 
			& \quad \textsf{[Elim]} \\
	&\frac{1}{\sqrt{2^{m+1}}}\sum_{y_0\in\F}\sum_{\y\in\F^m}
		e^{2\pi i \left(\frac{1}{4}y_0 + \frac{1}{2}y_0 Q(\x, \y) + R(\x, \y)\right)}\ket{f(\x,\y)} 
		&&\reduces \frac{1}{\sqrt{2^{m}}}\sum_{y\in\F^m}e^{2\pi i\left(\frac{1}{8} -\frac{1}{4}\lift{Q}(\x, \y) + R(\x, \y)\right)}\ket{f(\x,\y)} 
			& \textsf{[$\omega$]} \\
	&\frac{1}{\sqrt{2^{m+1}}}\sum_{y_0\in\F}\sum_{\y\in\F^m}
		e^{2\pi i \left(\frac{1}{2}y_0(y_i + Q(\x, \y)) + R(\x, \y)\right)}\ket{f(\x,\y)}
		&&\reduces \frac{1}{\sqrt{2^{m+1}}}\sum_{\substack{\y\in\F^m}}
			e^{2\pi i\left(R[y_i\gets \lift{Q}]\right)(\x, \y)}\ket{\left(f[y_i \gets Q]\right)(\x,\y)} 
			& \textsf{[HH]}
\end{alignat*}\end{minipage}}
\[
	\inference
		{P(\x,\y) = \frac{1}{4}y_ix + \frac{1}{2}y_i(y_j + Q(\x, \y)) + R(\x, \y) 
			= \frac{1}{4}y_j(1-x) + \frac{1}{2}y_j(y_i + Q'(\x, \y)) + R'(\x, \y)}
		{\frac{1}{\sqrt{2^{m+2}}}\sum_{\y\in\F^{m+2}}e^{2\pi i P(\x, \y)}\ket{f(\x,\y)} \reduces 
		 \frac{1}{\sqrt{2^{m}}}\sum_{\y\in\F^{m}}e^{2\pi i \left((1 - x)R[y_j\gets \lift{Q}] + 
		 	xR'[y_i\gets \lift{Q'}]\right)(\x, \y)}\ket{f(\x,\y)}}
	\tag*{\textsf{[Case]}}
\]
\caption{Path-sum reduction rules}\label{fig:rewrite}
\end{figure}

\begin{proposition}[Correctness]\label{thm:correctness}
	If $\xi \reducestrans \xi''$, then $\xi\equiv \xi'$.
\end{proposition}
The correctness of our rewrite system follows from direct calculation over symbolic exponential sums. As the proof is quite tedious, we leave it to \Cref{app:proof}. 

It is a trivial fact that our calculus is terminating, as every rule reduces the number of path variables. Moreover, each rewrite rule can be matched against in polynomial time, hence every path-sum reduces to a normal form in polynomial time.

\begin{proposition}[Strong normalization]\label{thm:normal}
	Every sequence of rewrites terminates with an irreducible path-sum. The sequence is linear in the number of path variables $m$ and for an $n$-qubit path-sum takes time polynomial in $n$ and $m$.
\end{proposition}

\subsection{Examples}

To illustrate our rewrite system, we give examples below. Further examples can be found in \Cref{app:examples}.

\begin{example}
Recall that the standard implementation of the Toffoli gate over Clifford+$T$ has the path-sum form \cite{amm14} \vspace{-1pt}
\[
	\textsf{Toffoli}_3:\ket{x_1x_2x_3}\mapsto \frac{1}{\sqrt{2^2}}\sum_{y_1, y_2\in\F} 
		e^{2\pi i \frac{1}{2}\left(x_3y_1 + x_1x_2y_1 + y_1y_2\right)}\ket{x_1x_2y_2}.
\]
We can verify that this is equivalent to the functional specification $\ket{x_1x_2x_3}\mapsto \ket{x_1x_2(x_3\oplus x_1x_2)}$ with the following sequence of reductions and algebraic manipulations:
\begin{align*}
	\ket{x_1x_2x_3} 
		&\mapsto \frac{1}{\sqrt{2^2}} \sum_{y_1, y_2\in\F} 
			e^{2\pi i \frac{1}{2}\left(x_3y_1 + x_1x_2y_1 + y_1y_2\right)}\ket{x_1x_2y_2}\quad \\
		&\mapsto \frac{1}{\sqrt{2^2}} \sum_{y_1, y_2\in\F} 
			e^{2\pi i \frac{1}{2}y_1(y_2 + x_3 + x_1x_2)}\ket{x_1x_2y_2} \\
		&\mapsto \frac{1}{\sqrt{2^2}} \sum_{y_2\in\F}\ket{x_1x_2(x_3\oplus x_1x_2)} \tag*{\textsf{[HH]}} \\
		&\mapsto \ket{x_1x_2(x_3\oplus x_1x_2)} \tag*{\textsf{[Elim]}}
\end{align*}
\end{example}

\begin{example}
The controlled-$T$ gate can be specified as the path-sum
\[
  \textsf{controlled-T}:\ket{x_1x_2} \mapsto e^{2\pi i \frac{x_1x_2}{8}}\ket{x_1x_2}.
\]
An implementation of the controlled-$T$ gate over Clifford+$T$ is given below:
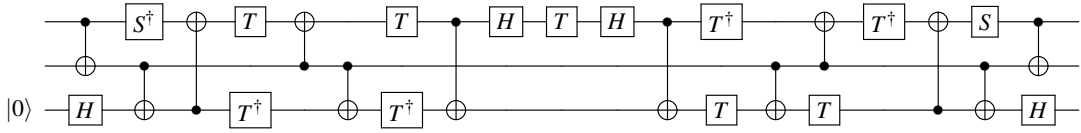
\[
\centerline{\footnotesize
\Qcircuit @C=1em @R=.7em {
 & \ctrl{1} & \gate{S^\dagger} & \targ & \gate{T} & \targ & \qw & \gate{T} & \ctrl{2} & \gate{H} & \gate{T} & \gate{H}
                  & \ctrl{2} & \gate{T^\dagger} & \qw & \targ & \gate{T^\dagger} & \targ & \gate{S} & \ctrl{1} & \qw  \\
 & \targ & \ctrl{1} & \qw & \qw & \ctrl{-1} & \ctrl{1} & \qw & \qw & \qw & \qw & \qw 
             & \qw & \qw & \ctrl{1} & \ctrl{-1} & \qw & \qw & \ctrl{1} & \targ &  \qw \\
 \lstick{\ket{0}} & \gate{H} & \targ & \ctrl{-2} & \gate{T^\dagger} & \qw & \targ & \gate{T^\dagger} & \targ & \qw & \qw & \qw
                            & \targ & \gate{T} & \targ & \gate{T} & \qw & \ctrl{-2} & \targ & \gate{H} & \qw
}
}
\]
Computing the canonical path-sum and reducing we get
\begin{align*}
	\ket{x_1x_2}\ket{0} 
		&\mapsto \frac{1}{\sqrt{2^4}}\sum_{\y\in\F^4}
			e^{2\pi i \frac{1}{8}\left( 4x_1x_2y_1 + 4x_1y_2 + 4y_1y_2 + y_2 + 4y_2y_3 + 4x_1x_2y_3 + 4x_1y_4 + 4y_3y_4 + 4x_1x_2\right)}
			\ket{x_1x_2y_4} \\
		&\mapsto \frac{1}{\sqrt{2^4}}\sum_{\y\in\F^4}
			e^{2\pi i \left(\frac{1}{2}y_1(y_2 + x_1x_2) + 
				\frac{1}{8}(4x_1y_2 + y_2 + 4y_2y_3 + 4x_1x_2y_3 + 4x_1y_4 + 4y_3y_4 + 4x_1x_2)\right)}
			\ket{x_1x_2y_4} \;\; \\
		&\mapsto \frac{1}{\sqrt{2^2}}\sum_{y_3, y_4\in\F}
			e^{2\pi i \frac{1}{8}(4x_1x_2 + x_1x_2 + 4x_1x_2y_3 + 4x_1x_2y_3 + 4x_1y_4 + 4y_3y_4 + 4x_1x_2)}
			\ket{x_1x_2y_4} \tag*{\textsf{[HH, Elim]}} \\
		 &\mapsto \frac{1}{\sqrt{2^2}}\sum_{y_3, y_4\in\F}
			e^{2\pi i \left(\frac{1}{2}y_3y_4 + \frac{1}{8}(x_1y_4 + x_1x_2)\right)}
			\ket{x_1x_2y_4} \\
		&\mapsto e^{2\pi i \frac{x_1x_2}{8}}\ket{x_1x_2}\ket{0} \tag*{\textsf{[HH, Elim]}}
\end{align*}
Hence the above circuit implements the controlled-$T$ gate, and provably leaves the ancilla clean.
\end{example}

\section{Completeness}\label{sec:completeness}

While our calculus computes a normal form in polynomial time, the normal forms are not necessarily unique\footnote{It was pointed out by an anonymous referee that uniqueness would imply that equivalence checking of reversible Boolean circuits is in P. As this problem is co-NP-complete, uniqueness of our normal forms would indeed imply $\text{P}=\text{co-NP}$.} and hence our reduction system is incomplete. For instance, the Clifford+$T$ identity
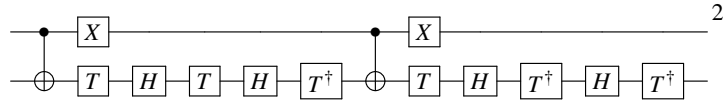
\[
\centerline{\footnotesize
\Qcircuit @C=1em @R=.7em {
 & \ctrl{1} & \gate{X} & \qw & \qw & \qw & \qw & \ctrl{1} & \gate{X} & \qw & \qw & \qw & \qw & \ustick{\;\;\;\;2}\qw  \\
 & \targ & \gate{T} & \gate{H} & \gate{T} & \gate{H} & \gate{T^\dagger} & \targ & \gate{T} & \gate{H} & \gate{T^\dagger}
 	& \gate{H} & \gate{T^\dagger} & \qw
}
}
\]
from \cite{bs16} gives the irreducible path-sum $\ket{x_1x_2} \mapsto \frac{1}{\sqrt{2^8}}\sum_{\y\in\F^8}e^{2\pi i \frac{1}{8}P(\x, \y)}\ket{x_1y_8}$ with phase polynomial
\begin{align*}
P(\x, \y) =
	2 &+ 6x_1x_2 + x_2 + y_1 + 4y_1(x_1 + x_2 + y_2) + 6y_2 + 4y_2y_3 + 2y_2x_1 + 3y_3 + 4y_3(x_1 + y_4) \\
	&+ 4y_4y_5 + 6y_4x_1 + y_5 + 4y_5(x_1 + y_6) + 6y_6 + 4y_6y_7 + 2y_6x_1 + 3y_7 + 4y_7(x_1 + y_8) + 7y_8.
\end{align*}
A complete verification procedure could proceed by explicitly expanding the values of remaining variables in the path-sum after all possible reductions have been made, and then checking equivalence to the identity transformation. In practice we found that this is generally not necessary, as our calculus, along with some additional observations, is sufficient to prove equivalence or non-equivalence for the majority of circuits. Moreover, these heuristics combined with path-sum reductions give a complete, polynomial-time procedure for determining equivalence of Clifford group circuits.

\subsection{Isometry restrictions}
Our first heuristic reduces the number of path variables in a \emph{well-formed} path sum when checking equivalence. Specifically, we denote by $\xi|_{f(\x,\y)=\x}$ the restriction of $\xi$ to solutions $\x\in\F^n, \y\in\F^m$ such that $f(\x,\y)=\x$, which we can write as the restricted sum
\[
	\ket{\x} \mapsto \frac{1}{\sqrt{2^m}}\sum_{\y\in\F^m, f(\x, \y) = \x}e^{2\pi i P(\x, \y)}\ket{\x}.
\]
Effectively, the sum $\frac{1}{\sqrt{2^m}}\sum_{\y\in\F^m, f(\x, \y) = \x}e^{2\pi i P(\x, \y)}$ gives the amplitude of the basis state $\ket{\x}$ in the output for a given input state $\ket{\x}$. If the path sum $\xi$ is well-formed (i.e. isometric), then this sum will be equal to $1$ exactly if $\xi$ is the identity transformation. We sum this up in the lemma below:
\begin{lemma}\label{lem:wf}
	Suppose $U_{\xi}:\ket{\x}\mapsto\frac{1}{\sqrt{2^m}}\sum_{\y\in\F^m}e^{2\pi i P(\x, \y)}\ket{f(\x, \y)}$
	is a well-formed path-sum. Then $\xi\equiv \ket{\x}\mapsto \ket{\x}$ if and only if 
	$\xi|_{f(\x, \y) = \x}\equiv \ket{\x}\mapsto \ket{\x}.$
\end{lemma}
Note that \cref{lem:wf} doesn't hold if $\xi$ is not well-formed, as $U_{\xi}$ may not be an isometry and so it may be that $U_{\xi}\ket{\x} = \ket{\x} + \ket{\psi}$ for some residual state $\ket{\psi}$. To reify the restricted path-sum $\xi|_{f(\x, \y)=\x}$ we find path variable substitutions which give $f_i(\x, \y)=x_i$ -- in particular, if for some index $i$ we have $f_i(\x, \y) = y_i\oplus Q(\x, \y)$ where $y_i$ doesn't appear in $Q(\x,\y)$, we can substitute $Q(\x, \y)$ for $y_i$ to get $f_i(\x, \y) = x_i$ and remove $y_i$ from the sum. Any restrictions which can't be reified are simply ignored. In practice this results in a significant simplification for some circuits, instantly removing up to $n$ path variables.

\subsection{Non-equivalence}
As the reduction rules of \cref{fig:rewrite} only suffice to prove \emph{positive} results, when no more reductions are possible we apply an observation that was found to be effective for proving that a path sum $\xi$ is \emph{not} the identity. In particular, recall that
\[
	\frac{1}{\sqrt{2^{m+1}}}\sum_{y_0\in\F}
	\sum_{\y\in\F^m}e^{2\pi i \left( \frac{1}{2}y_0Q(\x, \y) + R(\x,\y)\right)}\ket{f(\x, \y)} = 0 
\]
if $Q(\x, \y) = 1\mod 2$. If $Q$ is a non-zero Boolean-valued polynomial in \emph{only} input variables $x_i$, then there necessarily exists a solution $\x\in\Z^n$ such that $Q(\x) = 1\mod 2$ \cite{od14}, and in particular
\[
	\sum_{\y\in\F^m}e^{2\pi i \left( \frac{1}{2}y_0Q(\x, \y) + R(\x,\y)\right)}\ket{f(\x, \y)} = \frac{1}{\sqrt{2^{m+1}}}
	\sum_{\y\in\F^m}(1 - 1)e^{2\pi i R(\x,\y)}\ket{f(\x, \y)} = 0 .
\]
We sum this up in the following lemma.
\begin{lemma}\label{lem:negative}
	Suppose 
	$
		U_{\xi}: \ket{\x} \mapsto \frac{1}{\sqrt{2^{m+1}}}\sum_{y_0\in\F}
		\sum_{\y\in\F^m}e^{2\pi i \left( \frac{1}{2}y_0Q(\x, \y) + R(\x,\y)\right)}\ket{f(\x, \y)}
	$
	where $Q$ is a non-zero integer-valued polynomial not containing any path variables. Then $\xi \centernot{\equiv} \ket{\x}\mapsto\ket{\x}$.
\end{lemma}
Hence we can use a variant of \textsf{[HH]} where $Q$ contains only input variables to prove non-equivalence of a path-sum to the identity.

\subsection{Clifford completeness}

We can now show that together with the above simplifications, our path-sum reductions are complete for proving equivalence of Clifford group circuits. Recall that over the Clifford group, the path-sum interpretation of a circuit has phase polynomial of order at most $2$. Our proof of completeness rests on the fact that progress can always be made for an identity path-sum with only internal path variables and second-order phase polynomial, as shown below.

\begin{lemma}[Clifford progress \& preservation]\label{lem:pp}
	If $\xi$ is a path-sum such that $\xi\equiv\ket{\x}\mapsto\ket{\x}$, $\ord{P} \leq 2$ and $\xi$ contains only internal path variables, 
		then there exists $\xi'$ such that $\xi\reduces \xi'$ and $\ord{P'}\leq 2$.
\end{lemma}
\begin{proof}
	Since $P$ is at most second-order, we can write $P = y_0Q + R$
	for some internal path variable $y_0$ and polynomials $Q, R$ where $Q$ is at most first-order, and in particular has the form
	\[
		a\frac{1}{4} + b\frac{1}{2}Q'
	\]
	where $a,b\in\F$ and $Q'$ is a linear Boolean-valued polynomial. 
	We have 3 cases to consider, corresponding to the \textsf{[Elim]}, \textsf{[HH]} and \textsf{[$\omega$]} rules respectively.
	\paragraph{Case 1: $a=b=0$.}
		The variable $y_i$ does not appear in $P$, hence $\xi\reduces_{\textsf{[Elim]}} \xi'$ and $\ord{P'}=\ord{P} \leq 2$.
	\paragraph{Case 2: $a=0, b=1$.}
		If the polynomial $Q'$ contains a path variable $y_i$, then $Q' = y_i + Q''$ and $\xi\reduces_{\textsf{[HH]}} \xi'$. 
		Further, by \cref{lem:order}, $\ord{R[y_i\gets \lift{Q''}]} \leq \ord{R} \leq 2$ and $\xi'$ has only internal paths since $y_i\notin f$.
		
		If on the other hand $Q'$ only contains input variables, by \cref{lem:negative} $\xi \centernot{\equiv} \ket{\x} \mapsto \ket{\x}$,
		  a contradiction.
	\paragraph{Case 3: $a=1$.}
		The sum matches the left hand side of \textsf{[$\omega$]}, hence $\xi\reduces_{\textsf{[$\omega$]}} \xi'$.
		Further, by \cref{lem:order} 
		\[
			\ord{P'}=\ord{\frac{1}{8} - \frac{1}{4}\lift{Q''} + R} = \max \left\{ \ord{\frac{1}{8}}, \ord{\frac{1}{4}\lift{Q''}}, \ord{R}\right\} = 2.
		\]
\end{proof}

\begin{corollary}
	If $C$ is a Clifford-group quantum circuit, then $\sop{C} \equiv \ket{\x}\mapsto \ket{\x}$ can be decided in time polynomial in the space-time volume of $C$.
\end{corollary}
\begin{proof}
	Since $\sop{C}$ is well-formed, by \cref{lem:wf} it suffices to check $\sop{C}|_{f(\x,\y) = \x}\equiv \ket{\x}\mapsto \ket{\x}$. Further, as $f(\x,\y)$ is linear, we can compute via Gaussian elimination a solution $\y$ so that $f(\x,\y) = \x$ for any $\x$ -- if no such solution exists, $\sop{C}\centernot{\equiv} \ket{\x} \mapsto \ket{\x}$. Since each $f_i$ is linear, $\ord{P[y_i\gets f_i]}\leq\ord{P}\leq 2$, hence by \cref{lem:pp} and \cref{thm:normal}, either $\sop{C}|_{f(\x,\y) = \x}$ reduces to $\ket{\x}\mapsto\ket{\x}$ in polynomial-time or $\xi'\centernot{\equiv}\ket{\x}\mapsto\ket{\x}$.
\end{proof}

\section{Case studies}\label{sec:experiments}

We implemented our framework and verification algorithm in the open-source Haskell library \href{https://github.com/meamy/feynman}{\sc{Feynman}}. To test the efficacy of our methods, we performed verification of circuit optimizations (both correct and incorrect), as well as the verification of circuit implementations against formal path-sum specifications. All experiments were run in Debian Linux running on a quad-core 64-bit Intel Core i7 2.40 GHz processor and 8 GB RAM, and can be executed from the command line with \texttt{./feyn VerBench} and \texttt{./feyn VerAlg} for the translation validation and algorithm benchmarks, respectively.

\subsection{Translation validation}

\begin{table}
\footnotesize
\caption{Translation validation results. $n$ lists the number of qubits, Path vars gives the number of path variables, and Clifford and $T$ give the number of respective gates. Times for positive and negative verification measure the time to prove equivalence or non-equivalence against the optimized circuit or the optimized circuit with one random gate removed, respectively. Benchmarks with no timing results ran out of memory.}
\label{tab:results}
\centering
\begin{tabular}[t]{lrrrrrr} \toprule
Algorithm & $n$ & Path vars & Clifford & $T$ & \multicolumn{2}{c}{Time (s)} \\ \cmidrule(l{4pt}r{4pt}){6-7}
 & &  & & & Positive & Negative \\ \midrule
Grover\_5 & 9 & 200 & 1515 & 490 & 0.973 & 0.988 \\
Mod 5\_4 & 5 & 12 & 66 & 44 & 0.005 & 0.028  \\ 
VBE-Adder\_3 & 10 & 20 & 167 & 94 & 0.026 & 0.028 \\
CSLA-MUX\_3 & 15 & 40 & 289 & 132 & 0.099 & 0.055 \\
CSUM-MUX\_9 & 30 & 56 & 638 & 280 & 0.270 & 0.270 \\ 
QCLA-Com\_7 & 24 & 74 & 1237 & 297 & 0.530 & 0.543 \\
QCLA-Mod\_7 & 26 & 164 & 1641 & 650 & 9.446 & 10.517 \\
QCLA-Adder\_10 & 36 & 100 & 627 & 400 & 0.674 & 0.683  \\ 
Adder\_8 & 24 & 160 & 1419 & 614 & 1.968 & 2.018  \\
RC-Adder\_6 & 14 & 44 & 322 & 124 & 0.080 & 0.090  \\ 
Mod-Red\_21 & 11 & 60 & 392 & 192 & 0.110 & 0.119  \\
Mod-Mult\_55 & 9 & 28 & 180 & 84 & 0.028 & 0.009  \\ 
Mod-Adder\_1024 & 28 & 660 & 4363 & 3006 & 21.362 & 21.588  \\
Cycle 17\_3 & 35 & 1366 & 9172 & 6694 & -- & -- \\ 
GF($2^4$)-Mult & 12 & 28 & 263 & 180 & 0.063 & 0.061  \\
GF($2^5$)-Mult & 15 & 36 & 393 & 286 & 0.143 & 0.141  \\
GF($2^6$)-Mult & 18 & 44 & 559 & 402 & 0.279 & 0.291 \\
GF($2^7$)-Mult & 21 & 52 & 731 & 560 & 0.501 & 0.527 \\
GF($2^8$)-Mult & 24 & 60 & 975 & 712 & 0.837 & 0.881 \\
GF($2^9$)-Mult & 27 & 68 & 1179 & 918 & 1.304 & 1.369 \\
GF($2^{10}$)-Mult & 30 & 76 & 1475 & 1110 & 1.958 & 0.327 \\
GF($2^{16}$)-Mult & 48 & 124 & 3694 & 2832 & 16.028 & 17.539 \\ 
GF($2^{32}$)-Mult & 96 & 252 & 14259 & 11296 & 430.883 & 436.521 \\
GF($2^{64}$)-Mult & 192 & 508 & 55408 & 45120 & -- & -- \\
Hamming\_15 (low) & 17 & 76 & 612 & 158 & 0.367 & 0.168 \\
Hamming\_15 (med) & 17 & 184 & 1251 & 762 & 1.390 & 1.430 \\
Hamming\_15 (high) & 20 & 716 & 5332 & 3462 & 24.360 & 24.303 \\
HWB\_6 & 7 & 52 & 369 & 180 & 0.200 & 0.207  \\
HWB\_8 & 12 & 2282 & 17583 & 8895 & -- & -- \\
QFT\_4  & 5 & 84 & 218 & 136 & 0.084 & 0.089 \\ 
$\Lambda_3(X)$ & 5 & 12 & 52 & 36 & 0.004 & 0.011 \\
$\Lambda_3(X)$ (Barenco) & 5 & 12 & 66 & 44 & 0.007 & 0.046  \\
$\Lambda_4(X)$ & 7 & 20 & 87 & 58 & 0.009 & 0.008 \\
$\Lambda_4(X)$ (Barenco) & 7 & 20 & 127 & 84 & 0.014 & 0.024 \\
$\Lambda_5(X)$ & 9 & 18 & 112 & 80 & 0.015 & 0.017 \\
$\Lambda_5(X)$ (Barenco) & 9 & 28 & 160 & 124 & 0.030 & 0.031  \\
$\Lambda_{10}(X)$ & 19 & 68 & 297 & 190 & 0.110 & 0.111 \\ 
$\Lambda_{10}(X)$ (Barenco) & 19 & 68 & 493 & 324 & 0.219 & 0.210 \\ \bottomrule
\end{tabular}
\end{table}

Translation validation is an important tool for verifying that the transformations a compiler performs do not change the semantics of an input program. While it is generally desirable to prove that a compiler operates correctly on \emph{all} input programs, as with \emph{verified compilers} like \textsf{CompCert} \cite{l06} or \textsc{ReVerC} \cite{ars17} in the reversible domain, in many cases this is infeasible since the best optimizations are typically difficult to formally verify.

We used our algorithm to verify a suite of optimized benchmark circuits against their original input. For the optimization algorithm we chose the \textsc{GraySynth} algorithm from \cite{aam17} which is implemented in \textsc{Feynman} and verified each benchmark reported in that paper. \Cref{tab:results} reports the results of our experiments. All but $3$ of the benchmark circuits were successfully verified, with the remaining $3$ benchmarks running out of memory with a 6 GB limit. The high memory usage may be mitigated in the future by switching to a linear-space representation of the phase polynomial. The largest (completed) benchmark GF($2^{32}$), containing $96$ bits, $252$ path variables and over $25000$ gates completed in under 10 minutes, with the remainder all taking under a minute. 

To test the algorithm's ability to prove \emph{non-equivalence}, we also performed the verification of the optimized benchmark circuits after removing a randomly selected gate. Again, all but $3$ benchmarks were proven to be not equivalent, with the negative verification results taking about the same amount of time as positive results.

\subsection{Verifying quantum algorithms}

\begin{table}
\footnotesize
\caption{Results of verifying formally specified quantum algorithms.}
\label{tab:specs}
\centering
\begin{tabular}[t]{lrrrrrr} \toprule
Algorithm & $n$ & Path vars & Clifford & $T$ & \multicolumn{2}{c}{Time (s)} \\ \cmidrule(l{4pt}r{4pt}){6-7}
 & & & & & Positive & Negative \\ \midrule
\textsf{Toffoli}$_{50}$ & 97 & 190 & 855 & 665 & 1.084 & 1.064 \\
\textsf{Toffoli}$_{100}$ & 197 & 390 & 1755 & 1365 & 5.566 & 5.275 \\
\textsf{Maslov}$_{50}$ & 74 & 192 & 481 & 384 & 0.801 & 0.778 \\
\textsf{Maslov}$_{100}$ & 149 & 392 & 981 & 784 & 3.987 & 3.983 \\
\textsf{Adder}$_8$ & 40 & 56 & 334 & 196 & 0.142 & 0.143 \\
\textsf{Adder}$_{16}$ & 80 & 120 & 710 & 420 & 25.527 & 92.607 \\
\textsf{QFT}$_{16}$ & 16 & 16 & 256 & -- & 1.250 & 1.335 \\
\textsf{QFT}$_{31}$ & 31 & 31 & 961 & -- & 16.929 & 15.295 \\
\textsf{Hidden Shift}$_{20, 4}$ & 20 & 60 & 5254 & 56 & 1.067 & 0.862 \\
\textsf{Hidden Shift}$_{40, 5}$ & 40 & 120 & 6466 & 70 & 3.383 & 2.826 \\
\textsf{Hidden Shift}$_{60, 10}$ & 60 & 180 & 12784 & 140 & 13.217 & 12.351  \\
\textsf{Symbolic Shift}$_{20, 4}$ & 40 & 60 & 5296 & 56 & 1.859 & 1.849 \\
\textsf{Symbolic Shift}$_{40, 5}$ & 80 & 120 & 6638 & 70 & 6.953 & 7.905 \\
\textsf{Symbolic Shift}$_{60, 10}$ & 120 & 180 & 12804 & 140 & 35.583 & 29.614 \\ \bottomrule
\end{tabular}
\end{table}

To evaluate our framework as a tool for functional specification and verification, we implemented and verified several quantum algorithms (both without and with errors) directly against their specification as a path sum. \Cref{tab:specs} reports the results of our experiments, and we describe the algorithms and implementations below.

\paragraph{Reversible functions}

We implemented and verified a number of known algorithms for reversible functions. In particular, we performed verifications of Clifford+$T$ implementations of the generalized Toffoli and (out-of-place) addition functions,
\begin{align*}
	\textsf{Toffoli}_n &: \ket{x_1x_2\dots x_n}\mapsto \ket{x_1x_2\dots (x_n\oplus x_1x_2\dots x_{n-1})}, \\
	\textsf{Adder}_n &: \ket{\x}\ket{\y}\ket{\0}\mapsto \ket{\x}\ket{\y}\ket{\x+\y}
\end{align*}
We chose two implementations of the $n$-bit Toffoli gate -- using the standard decomposition into $2(n-3) + 1$ Toffoli gates and $n-3$ ancillas, and the Maslov decomposition \cite{m16} using \emph{relative phase} Toffolis and $\lceil \frac{n-3}{2}\rceil$ ancillas. For either implementation we were able to verify up to $100$ bit Toffoli gates in just seconds.

For the addition circuit, we used a standard out-of-place ripple-carry adder which uses $n-1$ ancilla bits to store intermediate carry values and an additional $n$ bit register to store the output, before copying out and uncomputing. The resulting circuit uses $5n - 1$ bits of space for an $n$ bit adder, and $4(n-1)$ Toffoli gates, which are then expanded to the Clifford+$T$ gate set. The specification itself was generated by implementing binary addition on symbolic vectors, and could ostensibly be classically tested to verify its own correctness. In this case, the size of the bitwise expansion of $\x+\y$ made it difficult to push to implementation sizes (e.g., 32 bits), though smaller sizes such as $16$ bits were verifiable within a minute. \emph{Relational} techniques -- e.g., representing the outputs of a path-sum as ``primed'' variables along with equations relating them -- may help to push verification of such functions to larger sizes.

\paragraph{The quantum Fourier transform}

To test our verification method against circuits using higher-order rotations, we verified an implementation of the quantum Fourier transform. We use a circuit from \cite{klm07} together with a final qubit permutation correction and verified it against the specification
\[
	\textsf{QFT}_n: \ket{\x}\mapsto \frac{1}{\sqrt{2^n}}\sum_{\y\in\F^n}e^{2\pi i \frac{\int{\x\cdot \y}}{2^n}}\ket{\y}.
\]
The phase polynomial $\int{\x \cdot \y}$ was generated in the obvious way -- by computing $\int{\x} = x_1 + 2x_2 + \dots + 2^{n-1}x_n$ and multiplying the polynomials. In this case our implementation was able to verify implementations up to $31$ bits in size, after which integer overflow occurs due to our handling of dyadic arithmetic. Given that the 31 bit implementation took only 16 seconds to verify, it appears that with better methods for handling dyadic arithmetic much larger sizes of the QFT are likely verifiable.

\paragraph{The quantum hidden shift algorithm}

To test our framework on more general quantum algorithms, we implemented a version of the quantum hidden shift algorithm \cite{r10} which has been previously used to test quantum simulation algorithms \cite{bg16}. In particular, given oracles $O_{f'}:\ket{\x}\mapsto f(\x+\vec{s})\ket{\x}$ and $O_{\tilde{f}}:\ket{\x}\mapsto \tilde{f}(\x)\ket{\x}$ for the shifted and dual bent functions $f', \tilde{f}:\F^n\rightarrow \{-1, +1\}$ respectively, the circuit $H^{\otimes n}O_{\tilde{f}}H^{\otimes n}O_{f'}H^{\otimes n}$ is known \cite{r10} to implement the mapping $\ket{\0}\mapsto \ket{\vec{s}}$.

Following \cite{bg16}, we generated random instances of Maiorana McFarland bent functions by setting $f'(\x,\y) = f((\x,\y)+\vec{s})=(-1)^{g(\x) + \x\y}$ with dual $\tilde{f}(\x,\y)=(-1)^{g(\y) + \x\y}$ for a random $\frac{n}{2}$ bit Boolean function $g$ of degree $3$. The circuit for $f$ is generated by, for a given number of alternations $A$, alternating between selecting 200 random $Z$ and controlled-$Z$ gates, then a random doubly controlled-$Z$ gate, expanded out to Clifford+$T$. We implemented two versions of the algorithm, one where a concrete shift is given by a randomly generated Boolean vector, and another where the shift is supplied symbolically via a quantum register. In the former case we verify the circuit for a given shift $\vec{s}$ against the specification $\ket{\vec{0}}\mapsto\ket{\vec{s}}$, and in the latter case we verify the specification $\ket{\vec{0}}\ket{\vec{s}} \mapsto \ket{\vec{s}}\ket{\vec{s}}$. \Cref{fig:shift} shows both circuits.

\begin{figure}
\begin{subfigure}{0.48\textwidth}
\[
\centerline{
\footnotesize
\Qcircuit @C=.7em @R=.7em {
 & \multigate{1}{H} & \multigate{1}{X^{\vec{s}}} & \ctrl{1} & \gate{O_g} & \multigate{1}{X^{\vec{s}}}
 	& \multigate{1}{H} & \ctrl{1} & \qw & \multigate{1}{H} & \qw \\
 & \ghost{H} & \ghost{X^{\vec{s}}} & \control \qw & \qw & \ghost{X^{\vec{s}}} 
 	& \ghost{H} & \control \qw & \gate{O_g} & \ghost{H} & \qw
}
}
\vspace{3mm}
\]
\subcaption{Hidden shift with a fixed shift $\vec{s}$.}
\end{subfigure}
\begin{subfigure}{0.48\textwidth}
\[
\centerline{
\footnotesize
\Qcircuit @C=.7em @R=.7em {
 & \multigate{1}{H} & \multigate{1}{X} & \ctrl{1} & \gate{O_g} & \multigate{1}{X}
 	& \multigate{1}{H} & \ctrl{1} & \qw & \multigate{1}{H} & \qw \\
 & \ghost{H} & \ghost{X} & \control \qw & \qw & \ghost{X}
 	& \ghost{H} & \control \qw & \gate{O_g} & \ghost{H} & \qw \\
  \lstick{\ket{\vec{s}}} & \qw & \ctrl{-1} & \qw & \qw & \ctrl{-1} & \qw & \qw & \qw & \qw & \qw
}
}
\]
\subcaption{Hidden shift with a symbolic shift.}
\end{subfigure}
\caption{Circuits for the Quantum Hidden Shift algorithm.}
\label{fig:shift}
\end{figure}

Our verification algorithm actually found a bug in our first implementation, which was a direct implementation of the circuit given in \cite{bg16}. After reimplementing the circuit based on \cite{r10}, we were able to verify both versions of the hidden shift algorithm for sizes exceeding those simulated in \cite{bg16} with only a fraction of the time (seconds versus hours \cite{bg16}). Our calculus further finds the correct output $\ket{\vec{s}}$ or $\ket{\vec{s}}\ket{\vec{s}}$ even without providing the specification, effectively simulating the algorithm rather than verifying it. Moreover, our implementation is deterministic compared to theirs which is probabilistic and only samples the output distribution, rather than compute it outright. It is interesting to note that their algorithm also uses a similar technique of effectively evaluating the circuit's phase polynomial -- however, by including the $T$ gate phases directly in the polynomial and solving \emph{around} them, rather than pushing them into state preparations, we save a massive amount of time for this algorithm. An interesting question for future research is to determine whether there are quantum algorithms which can be simulated more efficiently by their methods.

\section{Conclusion}

We have described a framework for the representation of partial isometries as sums over a discrete set of paths. As an alternative to matrices, our path-sums admit a symbolic representation using polynomials, for which there exists fixed-parameter polynomial size representations of Clifford+$R_k$ circuits. This allows the efficient computation and representation of the action of such a quantum circuit on an arbitrary basis state. Further, we have given a system of rewrite rules which can be used to reduce path-sums and perform functional verification. Our experiments have shown this to be a powerful framework for verifying large quantum circuits, particularly against formal mathematical specifications of quantum algorithms.

The work we have described here is only a preliminary step towards a fully-automated system of formal specification and verification for quantum circuits, and as such there are many issues for future work to address. One particularly appealing direction is to expand the path-sum framework to more general quantum programs, and to give a concrete syntax so that modular libraries of verified programs may be developed and used. Improvements can be made on the algorithmic side, from using Fourier expansions and relational methods to more efficiently store path-sums, to the use of algebraic decision diagrams or other mathematical tools to complete verification once no more reductions can be made. Another interesting direction, motivated by our experience writing path-sum proofs ``by hand,'' is to implement our framework in an interactive proof assistant, allowing inductive and higher-order proofs over entire families of quantum circuits.


\section{Acknowledgements}

The author wishes to thank Neil J. Ross and Michele Mosca for stimulating discussions on the topic of path sums, as well as the anonymous referees for their helpful comments on an earlier version. This work was supported in part by Canada's NSERC and CIFAR.


\bibliographystyle{eptcs}
\bibliography{paper}


\appendix
\section{Correctness of rewrite rules}\label{app:proof}

In this appendix we prove correctness for the rewrite rules of \cref{fig:rewrite}.

\begin{proof}[Proof of \cref{thm:correctness}.]
We verify each rewrite rule by direct calculation. Recall that by \cref{lem:poly}, for any Boolean-valued polynomial $Q$, $\lift{Q}(\x, \y) = Q(\x, \y)\mod 2$. 

\begin{align*}
	\textsf{[Elim]:}\qquad \frac{1}{\sqrt{2^{m+2}}}\sum_{y_0\in\F}\sum_{\y\in\F^m}e^{2\pi iP(\x, \y)}\ket{f(\x, \y)}
		&=\frac{1}{\sqrt{2^{m+2}}}\sum_{\y\in\F^m}(1+1)e^{2\pi iP(\x, \y)}\ket{f(\x, \y)} \\
		&=\frac{1}{\sqrt{2^{m}}}\sum_{\y\in\F^m}e^{2\pi iP(\x, \y)}\ket{f(\x, \y)}
\end{align*}

\begin{align*}
	\textsf{[$\mathbf{\omega}$]:}\qquad \frac{1}{\sqrt{2^{m+1}}}\sum_{y_0\in\F}&\sum_{\y\in\F^m} 
		e^{2\pi i \left(\frac{1}{4}y_0 + \frac{1}{2}y_0 Q(\x, \y) + R(\x, \y)\right)}\ket{f(\x, \y)} \\
		&= \frac{1}{\sqrt{2^{m+1}}}\sum_{\y\in\F^m}
			\left(1 + e^{2\pi i\left(\frac{1}{4} + \frac{1}{2}Q(\x, \y)\right)}\right)e^{2\pi i R(\x, \y)}\ket{f(\x, \y)}  \\
		&=\begin{cases} \frac{1}{\sqrt{2^{m+1}}}\sum_{\y\in\F^m}
			\left(1 + i\right)e^{2\pi i R(\x, \y)}\ket{f(\x, \y)} & \text{ if } Q(\x, \y) = 0\mod 2 \\
			 \frac{1}{\sqrt{2^{m+1}}}\sum_{\y\in\F^m}
			\left(1 - i\right)e^{2\pi i R(\x, \y)}\ket{f(\x, \y)} & \text{ if } Q(\x, \y) = 1\mod 2
			\end{cases} \\
		&=\begin{cases} \frac{1}{\sqrt{2^{m}}}\sum_{\y\in\F^m}
			e^{2\pi i \left( \frac{1}{8} + R(\x, \y)\right)}\ket{f(\x, \y)} & \text{ if } Q(\x, \y) = 0\mod 2 \\
			 \frac{1}{\sqrt{2^{m}}}\sum_{\y\in\F^m}
			e^{2\pi i \left( \frac{1}{8} + \frac{3}{4} + R(\x, \y)\right)}\ket{f(\x, \y)} & \text{ if } Q(\x, \y) = 1\mod 2
			\end{cases} \\
		&= \frac{1}{\sqrt{2^{m}}}\sum_{\y\in\F^m}
			e^{2\pi i \left( \frac{1}{8} + \frac{3}{4}\lift{Q}(\x, \y) +  R(\x, \y)\right)}\ket{f(\x, \y)}
\end{align*}

\begin{align*}
	\textsf{[HH]:}\qquad \frac{1}{\sqrt{2^{m+1}}}\sum_{y_0\in\F}&\sum_{\y\in\F^m}
		e^{2\pi i \left(\frac{1}{2}y_0(y_i + Q(\x, \y)) + R(\x, \y)\right)}\ket{f(\x, \y)} \\
		&= \frac{1}{\sqrt{2^{m+1}}}\sum_{\y\in\F^m}
			\left(1 + e^{2\pi i (y_i + Q(\x, \y))}\right)e^{2\pi i R(\x, \y)}\ket{f(\x, \y)} \\
		&= \frac{1}{\sqrt{2^{m+1}}}\sum_{\y\in\F^m,\; y_i=Q(\x, \y) \mod 2}
			\left(1 + e^{2\pi i (2k)}\right)e^{2\pi i R(\x, \y)}\ket{f(\x, \y)} \\
			& + \frac{1}{\sqrt{2^{m+1}}}\sum_{\y\in\F^m, \; y_i = 1 + Q(\x, \y) \mod 2}
			\left(1 + e^{2\pi i (2k + 1)}\right)e^{2\pi i R(\x, \y)}\ket{f(\x, \y)} \\
		&= 	\frac{2}{\sqrt{2^{m+1}}}\sum_{\y\in\F^m,\; y_i = Q(\x,\y) \mod 2}
			e^{2\pi i R(\x, \y)}\ket{f(\x, \y)} \\
		&= \frac{1}{\sqrt{2^{m+1}}}\sum_{\y\in\F^m}e^{2\pi i \left(R[y_i\gets \lift{Q}]\right)(\x, \y)}\ket{\left(f[y_i\gets Q]\right)(\x, \y)}
\end{align*}

\textsf{[Case]:}
Recall the precondition \[
P(\x,\y) = \frac{1}{4}y_ix + \frac{1}{2}y_i(y_j + Q(\x, \y)) + R(\x, \y) 
			= \frac{1}{4}y_j(1-x) + \frac{1}{2}y_j(y_i + Q'(\x, \y)) + R'(\x, \y).\]
\begin{align*}
	\frac{1}{\sqrt{2^{m+2}}}&\sum_{\y\in\F^{m+2}}e^{2\pi i P(\x, \y)}\ket{f(\x, \y)} \\
		&= \begin{cases}
			\frac{1}{\sqrt{2^{m+2}}}\sum_{\y\in\F^{m+2}}
				e^{2\pi i\left(\frac{1}{2}y_i(y_j + Q(\x, \y)) + R(\x, \y)\right)}\ket{f(\x, \y)}
				& \text{ if } x = 0 \\
			\frac{1}{\sqrt{2^{m+2}}}\sum_{\y\in\F^{m+2}}
				e^{2\pi i\left(\frac{1}{2}y_j(y_i + Q'(\x, \y)) + R'(\x, \y)\right)}\ket{f(\x, \y)}
				& \text{ if } x = 1
			\end{cases} \\
		&= \begin{cases}
			\frac{1}{\sqrt{2^{m}}}\sum_{\y\in\F^{m}}
				e^{2\pi i\left(R[y_j\gets \lift{Q}]\right)(\x, \y)}\ket{\left(f[y_j\gets Q]\right)(\x, \y)}
				& \text{ if } x = 0 \\
			\frac{1}{\sqrt{2^{m}}}\sum_{\y\in\F^{m}}
				e^{2\pi i\left(R'[y_i\gets \lift{Q'}]\right)(\x, \y)}\ket{\left(f[y_i\gets Q']\right)(\x, \y)}
				& \text{ if } x = 1
			\end{cases} \tag*{\text{by \textsf{[HH]} and \textsf{[Elim]}}} \\
		&= \frac{1}{\sqrt{2^{m}}}\sum_{\y\in\F^{m}}
				e^{2\pi i\left((1 - x)R[y_j\gets \lift{Q}] + xR'[y_i\gets \lift{Q'}]\right)(\x, \y)}\ket{f(\x, \y)}
				\qquad \text{since $y_i, y_j\notin f$}
\end{align*}
\end{proof}

\section{Reduction examples}\label{app:examples}

In this appendix we give further examples of the use of our reduction rules to prove circuit identities.

\begin{example}
To show the use of the \textsf{[$\omega$]} rule, we reduce the circuit $(SH)^3$ to the $\omega$ constant.

\begin{align*}
	\textsf{(SH)}^3: \ket{x} 
		&\mapsto \frac{1}{\sqrt{2}^3}\sum_{y_1, y_2, y_3\in\F}
    			e^{2\pi i \frac{1}{8}(4xy_1 + 6y_1 + 4y_1y_2 + 6y_2 + 4y_2y_3 + 6y_3 + 1)}\ket{y_3} \\
		&\mapsto \frac{1}{\sqrt{2}^3}\sum_{y_1, y_2, y_3\in\F}
    		    e^{2\pi i \left(\frac{1}{2}(\frac{1}{2}y_1 + y_1(y_2 \oplus 1 \oplus x)) + \frac{1}{8}(6y_2 + 4y_2y_3 + 6y_3 + 1)\right)}\ket{y_3} \\
		&\mapsto \frac{1}{\sqrt{2}^2}\sum_{y_2, y_3\in\F}
    			e^{2\pi i\frac{1}{8}(1 -2(y2 + 1 + x - 2y_2 - 2x - 2y_2x + 4y_2x) + 6y_2 + 4y_2y_3 + 6y_3 + 1)}\ket{y_3} 
    			\tag*{\textsf{[$\omega$]}} \\
		&\mapsto \frac{1}{\sqrt{2}^2}\sum_{y_2, y_3\in\F} e^{2\pi i\frac{1}{8}(2x + 4y_2x + 4y_2y_3 + 6y_3)}\ket{y_3} \\
		&\mapsto e^{2\pi i\frac{1}{8}(2x  + 6x)}\ket{x}  \tag*{\textsf{[HH, Elim]}} \\
		&\mapsto \omega\ket{x}.
\end{align*}
\end{example}

\begin{example}
The one-bit full adder has the reversible path-sum specification
\[
	\ket{x_1x_2x_3x_4} \mapsto \ket{x_1(x_1\oplus x_2)(x_1\oplus x_2\oplus x_3)(x_1x_2 \oplus x_1x_3 \oplus x_2x_3 \oplus x_4)}.
\] 
The implementation below over Clifford+$T$ was obtained by using the Reed-Muller decoding method of \cite{am16} to reduce the number of $T$ gates from the standard implementation using two Toffoli gates.
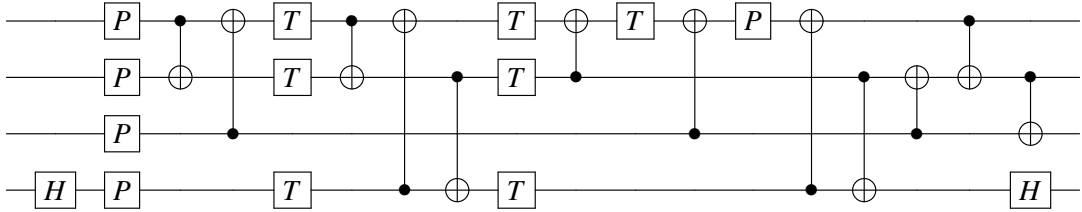
\[
\Qcircuit @C=1em @R=.7em { 
	& \qw & \gate{P} & \ctrl{1} & \targ       & \gate{T} & \ctrl{1} & \targ      & \qw       & \gate{T} & \targ       & \gate{T} & \targ      & \gate{P}
		& \targ & \qw & \qw & \ctrl{1} & \qw & \qw \\
	& \qw & \gate{P} & \targ      & \qw         & \gate{T} & \targ     & \qw        & \ctrl{2} & \gate{T} & \ctrl{-1} & \qw         & \qw        & \qw
		& \qw & \ctrl{2} & \targ & \targ & \ctrl{1} & \qw \\
	& \qw & \gate{P} & \qw        & \ctrl{-2} & \qw          & \qw       & \qw        & \qw       & \qw         & \qw         & \qw         & \ctrl{-2} & \qw
		& \qw & \qw & \ctrl{-1} & \qw & \targ & \qw \\
	& \gate{H} & \gate{P} & \qw        & \qw         & \gate{T} & \qw       & \ctrl{-3} & \targ    & \gate{T} & \qw         & \qw         & \qw        & \qw
		& \ctrl{-3} & \targ & \qw & \qw & \gate{H} & \qw
}
\]
We can verify that this circuit implements the one-bit adder specification as follows:
\begin{align*}
	 \ket{x_1x_2x_3x_4} 
		&\mapsto \frac{1}{\sqrt{2}^2}\sum_{y_1, y_2\in\F}
    			e^{2\pi i \frac{1}{2}(y_1y_2 + y_1x_1x_2 + y_1x_1x_3 + y_1x_2x_3 + y_1x_4)}\ket{x_1(x_1\oplus x_2)(x_1\oplus x_2\oplus x_3)y_2} \\
		&\mapsto \frac{1}{\sqrt{2}^2}\sum_{y_1, y_2\in\F}
    		    e^{2\pi i \frac{1}{2}y_1(y_2 + x_1x_2 + x_1x_3 + x_2x_3 + x_4)}\ket{x_1(x_1\oplus x_2)(x_1\oplus x_2\oplus x_3)y_2} \\
		&\mapsto \ket{x_1(x_1\oplus x_2)(x_1\oplus x_2\oplus x_3)(x_1x_2 \oplus x_1x_3 \oplus x_2x_3 \oplus x_4)} 
    				\tag*{\textsf{[HH, Elim]}}
\end{align*}
\end{example}

\end{document}

%% file: macros.tex
\theoremstyle{plain}
\newtheorem{theorem}{Theorem}[section]
\newtheorem{lemma}[theorem]{Lemma}
\newtheorem{proposition}[theorem]{Proposition}

\newtheorem{corollary}[theorem]{Corollary}

\theoremstyle{definition}
\newtheorem{definition}[theorem]{Definition}
\newtheorem{example}[theorem]{Example}

\theoremstyle{remark}
\newtheorem{remark}[theorem]{Remark}
\newtheorem*{notation*}{Notation}

\newcommand{\F}{\mathbb{Z}_2}
\newcommand{\Z}{\mathbb{Z}}

\newcommand{\D}{\mathbb{D}}

\newcommand{\ord}[1]{\mathsf{ord}\left(#1\right)}

\newcommand{\ket}[1]{|#1\rangle}

\newcommand{\cnot}{\mathrm{CNOT}}

\renewcommand{\vec}[1]{\mathbf{#1}}

\newcommand{\x}{\vec{x}}
\newcommand{\y}{\vec{y}}

\newcommand{\0}{\vec{0}}

\usepackage{todonotes}

\newcommand{\sop}[1]{\llbracket #1 \rrbracket}

\renewcommand{\int}[1]{\left[#1\right]}

\newcommand{\lift}[1]{\overline{#1}}

\newcommand{\reduces}{-->}
\newcommand{\reducestrans}{-->*}
